\documentclass{article}[12pt]

 \usepackage[margin=1.5in]{geometry}

\usepackage{amssymb}
\usepackage{amsfonts}
\usepackage{amsmath}
\usepackage{amsthm}
\usepackage{pifont}
\usepackage{nicefrac}
\usepackage{url}

\usepackage[usenames,dvipsnames]{pstricks}
\usepackage{epsfig}
\usepackage{pst-grad} 
\usepackage{pst-plot} 

\usepackage{algpseudocode}

\usepackage{comment}

\newtheorem{theorem}{Theorem}[section]

\newtheorem{proposition}[theorem]{Proposition}

\title{Manipulation and Control Complexity of Schulze Voting}

\author{Curtis Menton\footnote{Supported in part by grants NSF-CCF-0915792 and NSF-CCF-1101479.}\\
  Department of Computer Science \\
  University of Rochester \\
  Rochester, NY, USA
\and
Preetjot Singh\\
EECS,\\
Northwestern University,\\
Evanston, IL, USA
}

\begin{document}

\maketitle

\begin{abstract}
Schulze voting is a recently introduced voting system enjoying unusual
popularity and a high degree of real-world use, with users including
the Wikimedia foundation, several branches of the Pirate Party, and
MTV.  It is a Condorcet voting system that determines the winners of
an election using information about paths in a graph representation of
the election.  We resolve the complexity of many electoral control
cases for Schulze voting.  We find that it falls short of the best
known voting systems in terms of control resistance, demonstrating
vulnerabilities of concern to some prospective users of the system.
\end{abstract}



\section{Preliminaries}

\subsection{Schulze Voting} 
Schulze voting is a Condorcet voting system recently introduced by
Marcus Schulze~\cite{sch:j:schulze}.  It was designed to
effectively handle candidate cloning:  In many voting systems, the
inclusion of several similar candidates ends up spreading out their
support from similarly-minded voters and thus lessens the influence of
those voters.  
It has a somewhat more complex winner procedure than other common
voting systems, requiring the use of a graph best-path finding
algorithm, but it is still solvable in polynomial time, rendering
Schulze a tractable voting system.

Schulze voting is currently used by a number of organizations for
their internal elections, including the Wikimedia foundation, several
branches of the Pirate Party, a civil-liberties focused political
party, and by MTV to rank music videos.  This level of real-world
usage is unusual for a new, academically proposed and studied voting
system, but it makes Schulze voting more compelling to analyze.

As a typical Condorcet system, the winners are determined by examining
the pairwise contests between candidates.  We will thus
introduce some useful functions and notation.  The \emph{advantage
  function} is a function on pairs of candidates where $adv(a,b)$ is
the number of voters in the given election that prefer $a$ to $b$
(that is, rank $a$ higher than $b$ in their preferences).  The
\emph{net advantage function} gives the net difference in advantage
for one candidate over another.  We define the net advantage between
$a$ and $b$ to be the following: $netadv(a,b) = adv(a,b) - adv(b,a)$.

The winners in a Schulze election are determined as follows.  Generate
the net advantage scores for the election, and represent this data as
a graph, with vertices for candidates and directed edges denoting net
advantage scores.  Determine the strongest
paths in this graph between every pair of candidates by the following
metric: the weight of a path is the lowest weight edge in the path.
This is the ``bottleneck'' metric.  We can adapt the Floyd-Warshall
dynamic programming algorithm to find the weigh of such paths in
polynomial time~\cite{sch:j:schulze}.

\begin{figure}[h]
\begin{algorithmic}
\Procedure{Schulze Paths}{$netadv$}
\State initialize $paths$ to $netadv$
\For{$k = 1$ to $m$}
\For{$i = 1$ to $m$}
\For{$j = 1$ to $m$}
\If{$i == j$} 
\State next 
\EndIf
\State $newpath \gets min(paths[i][k], paths[k][j])$
\State $paths[i][j] \gets max(newpath, paths[i][j])$
\EndFor
\EndFor
\EndFor
\State \Return $paths$
\EndProcedure
\end{algorithmic}
\caption{Procedure to calculate Schulze best-path scores, adapted from
  Floyd-Warshall algorithm.}
\end{figure}

Once we have the best path
weights, we build another graph with the same vertices where there is
a directed edge from $a$ to $b$ if the best path from $a$ to $b$ is
better than the best path from $b$ to $a$.  The winners of the
election are the candidates with no in-edges in this final graph.

Note that a Condorcet winner will have positive net advantage edges to
every other candidate, and so they will easily have the better paths
to every other candidate and be the only winner of the election.

\subsubsection{Example Election}

Consider an election over candidates $\{a,b,c\}$ with the
following voters.  





\vspace{1em}
\begin{tabular}{c|c}
\# Voters & Preferences \\
\hline
3 & $a > b > c$ \\
3 & $b > c > a$ \\
2 & $c > a > b$ \\
\end{tabular}
\vspace{1em}

These voters express the following net advantage function over the
candidates.

\vspace{1em}
\begin{tabular}{c|ccc}
& $a$ & $b$ & $c$ \\
\hline
$a$ &    & 2  & -2  \\
$b$ &  -2 &   & 4 \\
$c$ &   2 &  -4&  \\
\end{tabular}
\vspace{1em}

We can also represent this election as a  graph as follows.


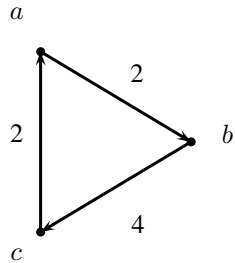
\begin{figure}[h]
\centering 
\scalebox{1} 
{
\begin{pspicture}(0,-1.7432812)(3.02875,1.7432812)
\psdots[dotsize=0.12](0.38,1.1017188)
\psdots[dotsize=0.12](0.38,-1.2982812)
\psdots[dotsize=0.12](2.38,-0.09828125)
\psline[linewidth=0.04cm,arrowsize=0.05291667cm
  2.0,arrowlength=1.4,arrowinset=0.4]{->}(0.38,1.1017188)(2.38,-0.09828125)
\psline[linewidth=0.04cm,arrowsize=0.05291667cm
  2.0,arrowlength=1.4,arrowinset=0.4]{->}(2.38,-0.09828125)(0.38,-1.2982812)
\psline[linewidth=0.04cm,arrowsize=0.05291667cm
  2.0,arrowlength=1.4,arrowinset=0.4]{->}(0.38,-1.2982812)(0.38,1.1017188)
\usefont{T1}{ptm}{m}{n}
\rput(0.063125,1.6117188){$a$}
\usefont{T1}{ptm}{m}{n}
\rput(2.8679688,0.01171875){$b$}
\usefont{T1}{ptm}{m}{n}
\rput(0.0578125,-1.5882813){$c$}
\usefont{T1}{ptm}{m}{n}
\rput(1.6634375,0.81171876){2}
\usefont{T1}{ptm}{m}{n}
\rput(1.675625,-1.1882813){4}
\usefont{T1}{ptm}{m}{n}
\rput(0.0634375,0.01171875){2}
\end{pspicture} 
}
\caption{Election graph for the example Schulze election.}
\end{figure}

Now we must find the best paths between pairs of candidates.  Note
that our three candidates are in a Condorcet cycle and there is no
clearly dominant candidate: each of them lose to one other
candidate, and each candidate has a path to both of the other candidates.  The
following are the Schulze path scores.

\vspace{1em}
\begin{tabular}{c|ccc}
& $a$ & $b$ & $c$ \\
\hline
$a$ &  & 2 & 2\\
$b$ & 2 &  & 4\\
$c$ & 2 & 2 & \\
\end{tabular}
\vspace{1em}

Candidates $a$ and $b$ have equally good paths to each other, as do
$a$ and $c$.  Candidate $b$ has a strictly stronger path to $c$ than
$c$ does to $b$.  So we can see that the winners will be $a$ and $b$,
as neither of these candidates are beaten in best-path strength.

\subsection{Election Manipulative Actions}

There are several classes of manipulative action problems studied in
computational social choice.  The primary classes of these problems
are \emph{manipulation}, where a voter or voters strategically vote to
affect the result of an election~\cite{bar-tov-tri:j:manipulating},
\emph{bribery}, where an outside briber pays off voters to change
their votes~\cite{fal-hem-hem:j:bribery}, and \emph{control}, where an
election organizer alters the structure of an election to change the
result~\cite{bar-tov-tri:j:control}.  We will later describe
manipulation and control in more detail.  

These problems are
formalized as decision problems in the way typical in computer science
and we study their complexity in various voting systems.  There is
some standard terminology about the behavior of voting systems under
manipulative actions.  A voting system is \emph{immune} to a
manipulative action if it can never change the result of an election
in the voting system, and it is \emph{susceptible} otherwise.  If a
voting system is susceptible to some action and the decision problem
is in P, then it is \emph{vulnerable} to it, while if the problem
is NP-hard, the voting system is \emph{resistant} to it.

The original versions of these problems are the \emph{constructive}
cases, where the goal is to make a particular candidate win.  The
alternative class of problems are the \emph{destructive} cases, where
the goal is to make a particular candidate not win.  These were
introduced by introduced by Conitzer et
al.~\cite{con-lan-san:j:few-candidates} in the context of
manipulation and later by Hemaspaandra et al.\ in the context of
control~\cite{hem-hem-rot:j:destructive-control}.  Though the
constructive goal may be more desirable, the destructive version of a
problem may be easy where the constructive version is not, so
investigating both is useful.

Additionally, there are differing versions of these problems based on
our tie-handling philosophy.  The work presented here follows the
nonunique winner model, where we claim success in our manipulative
action when the preferred candidate is one of possibly several winners
in the election in the constructive case, or is not among the set of
winners in the election in the destructive case.  In contrast, the
original paper exploring control~\cite{bar-tov-tri:j:control} used the
unique winner model, where the goal is to cause a candidate to be the
only winner of an election, or in the destructive cases for the
candidate to be not a unique winner. Both models are common in the
literature.

\subsection{Manipulation}

The manipulation problem is the most basic of the election
manipulative action problems.  The manipulation problem models
strategic voting, where a voter or set of voters attempt to vote in
some way (not necessarily reflecting their true preferences) to sway
the result of the election and cause their favorite candidate to win,
or else to cause some hated candidate to lose.  The computational
complexity of this problem was first studied by Bartholdi, Tovey, and
Trick~\cite{bar-tov-tri:j:manipulating}.  We will now formally define
the manipulation problem in the nonunique-winner model.

\begin{paragraph}{Manipulation}
\begin{description}
\item[Given] An election $(C,V)$, a set of manipulators $M$, and a
  distinguished candidate $p$.
\item[Question (Constructive)] Is there a way to assign the votes of
  $M$ such that $p$ is a winner of the election $(C, V \cup M)$?
\item[Question (Destructive)] Is there a way to assign the votes of
  $M$ such that $p$ is not a winner of the election $(C, V \cup M)$?
\end{description}
\end{paragraph}

\subsection{Control}

Control encompasses actions taken by an election chair to change the
structure of an election to achieve a desired result.  This can
include adding or deleting candidates or voters, or partitioning
either candidates or voters and performing initial subelections.  The
study of election control was initiated by Bartholdi, Tovey, and
Trick~\cite{bar-tov-tri:j:control}, and many subsequent papers have
investigated the complexity of control in various voting systems.

Much of the study of control has had the goal of finding voting
systems that are highly resistant to control.  Faliszewski et
al.\ showed that related systems Llull and Copeland voting are
resistant to every case of constructive
control~\cite{fal-hem-hem-rot:j:llull}.  Hemaspaandra et
al.\ constructed unnatural hybrid voting systems that resist every
case of control, proving that such systems can
exist~\cite{hem-hem-rot:j:hybrid}.  Erd{\'e}lyi et al.\ showed that
the system fallback voting resists 20 out of the 22 standard cases of
control, and it stands as the most resistant natural voting
system~\cite{erd-rot:c:fallback, erd-fel:c:param,
  erd-pir-rot:c:open-probs}.   No natural system resistant to every case
of control has yet been found.

The various control problems loosely model many real-world actions.
The cases of adding and deleting voters correspond to voter
registration drives and voter suppression efforts.  The cases of
adding and deleting candidates correspond to ballot-access procedures
that effectively remove many candidates from elections.  Cases of
partitioning voters are similar to the real-world practice of
gerrymandering (though that has additional geographic constraints),
and cases of partitioning candidates correspond to primary elections
or runoffs.  We will now formally define the various cases of control,
in the nonunique-winner model.

\paragraph{Control by Adding Candidates}
\begin{description}
\item[Given] Disjoint candidate sets $C$ and $D$, a voter set $V$ with
  preferences over $C \cup D$, a distinguished candidate $p \in C$,
  and $k \in \mathbb{N}$.
\item[Question (Constructive)] Is it possible to make $p$ a winner of
  an election $(C \cup D', V)$ with some $D' \subseteq D$ where
  $\|D'\| \leq k$?
\item[Question (Destructive)] Is it possible to make $p$ not a  winner of an election
  $(C \cup D', V)$ with some $D' \subseteq D$ where $\|D'\| \leq k$?
\end{description}

\paragraph{Control by Adding an Unlimited Number of Candidates}
\begin{description}
\item[Given] Disjoint candidate sets $C$ and $D$, a voter set $V$ with
  preferences over $C \cup D$, and a distinguished candidate $p \in
  C$.
\item[Question (Constructive)] Is it possible to make $p$ a winner of
  an election $(C \cup D', V)$ with some $D' \subseteq D$?

\item[Question (Destructive)] Is it possible to make $p$ not a winner
  of an election $(C \cup D', V)$ with some $D' \subseteq D$?
\end{description}

\paragraph{Control by Deleting Candidates}
\begin{description}
\item[Given] An election $E=(C,V)$, a distinguished candidate $p \in
  C$, and $k \in \mathbb{N}$.
\item[Question (Constructive)] Is it possible to make $p$ a  winner of an election
  $(C-C', V)$ with some $C' \subseteq C$ where $\|C'\| \leq k$?

\item[Question (Destructive)] Is it possible to make $p$ not a winner
  of an election $(C-C', V)$ with some $C' \subseteq (C - \{p\})$
  where $\|C'\| \leq k$?
\end{description}

\paragraph{Control by Adding Voters}
\begin{description}
\item[Given] A candidate set $C$, disjoint voter sets $V$ and $W$, a
  distinguished candidate $p \in C$, and
  $k \in \mathbb{N}$.
\item[Question (Constructive)] Is it possible to make $p$ a  winner of an
  election $(C, V \cup W')$ for some $W' \subseteq W$ where $\|W'\| \leq
  k$?

\item[Question (Destructive)] Is it possible to make $p$ not a  winner of an
  election $(C, V \cup W')$ for some $W' \subseteq W$ where $\|W'\| \leq
  k$?
\end{description}

\paragraph{Control by Deleting Voters}
\begin{description}
\item[Given] An election $E=(C,V)$, a distinguished candidate $p \in
  C$, and $k \in \mathbb{N}$.
\item[Question (Constructive)] Is it possible to make $p$ a winner of an
  election $(C, V - V')$ for some $V' \subseteq V$ where $\|V'\| \leq
  k$?
\item[Question (Destructive)] Is it possible to make $p$ not a  winner of an
  election $(C, V - V')$ for some $V' \subseteq V$ where $\|V'\| \leq
  k$?
\end{description}

\paragraph{Control by Partition of Candidates}
\begin{description}
\item[Given] An election $E=(C,V)$ and a distinguished candidate $p
  \in C$.
\item[Question (Constructive)] Is there a partition $(C_1, C_2)$ of $C$
  such that $p$ is a final winner of the election $(D \cup C_2, V)$,
  where $D$ is the set of candidates surviving the initial subelection
  $(C_1, V)$?

\item[Question (Destructive)] Is there a partition $C_1, C_2$ of $C$
  such that $p$ is not a final winner of the election $(D \cup C_2,
  V)$, where $D$ is the set of candidates surviving the subelection
  $(C_1, V)$?
\end{description}

\paragraph{Control by Runoff Partition of Candidates}
\begin{description}
\item[Given] An election $E=(C,V)$ and a distinguished candidate $p
  \in C$.
\item[Question (Constructive)] Is there a partition $C_1, C_2$ of $C$
  such that $p$ is a final winner of the election $(D_1 \cup D_2, V)$,
  where $D_1$ and $D_2$ are the sets of surviving candidates from the
  subelections $(C_1, V)$ and $(C_2, V)$?

\item[Question (Destructive)] Is there a partition $C_1, C_2$ of $C$
  such that $p$ is not a final winner of the election $(D_1 \cup D_2,
  V)$, where $D_1$ and $D_2$ are the sets of surviving candidates from
  the subelections $(C_1, V)$ and $(C_2, V)$?
\end{description}

\paragraph{Control by Partition of Voters}
\begin{description}
\item[Given] An election $E=(C,V)$ and a distinguished candidate $p
  \in C$.

\item[Question (Constructive)] Is there a partition $V_1, V_2$ of $V$
  such that $p$ is a final winner of the election $(D_1 \cup D_2, V)$
  where $D_1$ and $D_2$ are the sets of surviving candidates from the
  subelections $(C, V_1)$ and $(C, V_2)$?

\item[Question (Destructive)] Is there a partition $V_1, V_2$ of $V$
  such that $p$ is not a final winner of the election $(D_1 \cup D_2,
  V)$ where $D_1$ and $D_2$ are the sets of surviving candidates from
  the subelections $(C, V_1)$ and $(C, V_2)$?
\end{description}

\section{Results}

Parkes and Xia studied the complexity of the manipulative action
problems for Schulze voting~\cite{par-xia:c:ranked-pairs}.  They
proved resistance for Schulze voting for bribery and for some of the
control cases, and showed that constructive manipulation is easy with
a single manipulator.  Recently Gaspers et al.\ showed that
coalitional manipulation is easy as well~\cite{gas-kal-nar-wal:c:schulze}.
However several control cases as well as
constructive manipulation with multiple manipulators remained open.
We now will present our results on control and manipulation for
Schulze voting.  Our new results are the
following:

\begin{itemize}
\item In the nonunique-winner model, it is never necessary to make
  multiple manipulators vote differently from each other to
  successfully perform manipulation, while this does not hold for the
  unique-winner model.
\item Schulze voting is resistant to constructive control by unlimited
  adding of candidates.
\item Schulze voting is resistant to constructive control by deleting
  candidates.
\item Schulze voting is resistant to constructive control by
  partition/runoff partition of candidates, ties promote or ties
  eliminate.
\item Schulze voting is vulnerable to destructive control by
  partition/runoff partition of candidates, ties promote or ties
  eliminate.
\item Schulze voting is resistant to constructive or destructive control by
  partition of voters, ties promote or ties eliminate.  
\end{itemize}

Table~\ref{control-table} shows the behavior of Schulze voting, as
well as several other voting systems, under control.  



\begin{table*}[htp]
  \centering
  \begin{tabular}{|l|l||l|l|l|l|l|l|l|l|}
    \hline
    Control by & Tie & 
    \multicolumn{2}{|c|}{Plurality} & \multicolumn{2}{|c|}{Approval} &  \multicolumn{2}{|c|}{Fallback} &
    \multicolumn{2}{|c|}{Schulze} \\
    & Model & C& D & C& D & C& D & C& D \\
    \hline \hline
    Adding Candidates & & R & R & I & V  & R & R & R & S\\
    \hline
    Adding Candidates (unlimited)  & & R & R & I & V & R & R & \textbf{R} & S \\
    \hline
    Deleting Candidates & & R & R & V & I & R & R & \textbf{R} & S\\
    \hline
    Partition of Candidates & TE & R & R & V & I  & R & R & \textbf{R} & \textbf{V}\\
    \hline
    & TP & R & R & I & I  & R & R & \textbf{R} & \textbf{V} \\
    \hline
    Run-off Partition of Candidates & TE & R & R & V & I  & R & R & \textbf{R}    & \textbf{V}\\
    \hline
    & TP & R & R & I & I  & R & R & \textbf{R} & \textbf{V}\\
    \hline
    Adding Voters & & V & V & R & V & R & V & R & R\\
    \hline
    Deleting Voters & & V & V & R & V & R & V & R & R\\
    \hline
    Partition of Voters & TE & V & V & R & V & R & R & \textbf{R} & \textbf{R}\\
    \hline
    & TP & R & R & R & V & R & R & \textbf{R} & \textbf{R} \\
    \hline
  \end{tabular}
  \caption[Control behavior under Schulze voting and other voting
    systems for comparison.]{\label{control-table}Control behavior
    under Schulze voting and other voting systems for comparison.  V,
    R, S, and I stand for vulnerable, resistant, susceptible, and
    immune.  Results proved in this work in bold, other results 
    from~\protect\cite{erd-rot:c:fallback,erd-fel:c:param,erd-pir-rot:c:open-probs,hem-hem-rot:j:destructive-control,par-xia:c:ranked-pairs}.}
\end{table*}

\subsection{Manipulation}

Parkes and Xia proved that manipulation is easy for Schulze voting
with a single manipulator~\cite{par-xia:c:ranked-pairs}.  Their
algorithm was designed for the unique-winner model but as they note it
is easily adaptable to the nonunique-winner model as well.  It is not,
however, easily adaptable to cases with multiple manipulators.  With a
single manipulator, it can be seen whether manipulation is possible
simply by checking if the relative Schulze scores between $p$ and
other candidates are such that it is feasible.  That is, in the
unique-winner model, checking that no other candidate is beating $p$
in Schulze score, or in the nonunique-winner model, checking that no
candidate is beating $p$ in Schulze score by more than two.  With two
or more manipulators this is not the case in either tie-handling
model.


In the nonunique winner model, in positive manipulation instance all
manipulators in the coalition can always vote identically, while in
the unique model they may sometimes have to vote differently to
succeed.  We will demonstrate the first point.  

By including a single manipulator vote, we shift all of the net
advantage values in the election by one in either direction.  With $m$
manipulators, if they are given the same votes, all of the net
advantage values will shift by $m$.  If we instead give manipulators
differing votes, some or all of the net advantage scores will change
by less than $m$.  So we would have to assign the manipulator votes
differently if there was some pair of candidates where we needed to
alter their net advantage by less than $m$.  We will show that that
would never be necessary.  

Suppose we have a positive Schulze manipulation instance with $p$ as
the preferred candidate and with $m$ manipulators, and there are
candidates $b$ and $c$ such that we must alter $netadv(b,c)$ by less
than $m$ to make $p$ a winner.  We will be prevented from decreasing
$netadv(b,c)$ by $m$ if doing so will excessively weaken a path from
$p$ to some $d$.  This will be the case if $(b,c)$ is on the best path
from $p$ to $c$, $p$ has a path to $b$ that after manipulation is
greater than $netadv(b,c) - m$ in strength, and there is an equally
strong path from $c$ to $d$, but also there is a path from $c$ to
$p$ that will be greater than $netadv(b,c) - m$ in strength.

We would be prevented from increasing $netadv(b,c)$ by $m$ if doing so
would strengthen a path from some $a$ to $p$ to be greater than
$netadv(p,a)$ after manipulation. This would happen if $(b,c)$ is a
bottleneck on the strongest path from $a$ to $p$, with other edges
stronger as well.  The strongest path from $c$ to $p$ would have to be
at least $netadv(b,c) + m$ in strength after manipulation.  But if
this is the case, with $(b,c)$ also being a bottleneck from $p$ to
some candidate, and therefore with the best path to $c$ containing
$(b,c)$, $p$ would lose to $c$ if we do not increase $(b,c)$ by $m$.
But we assumed we could make $p$ a winner by changing $netadv(b,c)$ by
less than $m$, so this is a contradiction.


Schulze manipulation in the unique-winner model does not have this
property and there are instances where it is necessary to have
manipulators vote differently.  Consider the following instance with
two manipulators.  Again we will have candidates $b$ and $c$ such that
$(b,c)$ is the edge of interest in the election graph.  A candidate
$a$ has a strong path to $b$, and has $(b,c)$ as a bottleneck edge in
their best path to a candidate $d$.  $d$ has a path to $b$ of weight
$netadv(b,c)$.  $a$ is performing strongly against all other
candidates, such that all but $d$ can at best tie them in Schulze
score after manipulation.  The preferred candidate $p$ needs
$netadv(b,c)$ to stay at least where it is so they can beat $c$, but
increasing it will not cause any problems.  So now the problem is to
make $p$ not just a winner but a unique winner.  We can't decrease
$netadv(b,c)$ without making $p$ lose, but if we increase it, $a$ will be at least tied with
every other candidate and they will also be a winner.  The only way to
make $p$ a unique winner is to keep $(b,c)$ unchanged while boosting
the path from $d$ to $a$ such that $a$ loses to $d$.  So to make $p$ a
unique winner, we would have to have one manipulator rank $b$ over $c$
while the other ranks $c$ over $b$.

\subsection{Control}

\subsubsection{Constructive Control by Adding/Unlimited Adding of Candidates
  Case}

\begin{theorem}Schulze voting is resistant to constructive control by
    adding of candidates and constructive control by unlimited adding
    of candidates.
\end{theorem}

\begin{proof}
We will prove this case through a reduction from 3SAT.  Given a 3SAT
instance $(U, Cl)$, we will construct a control instance $(C,D,p,k)$
as follows.  The original candidate set $C$ will be the following: 

\begin{itemize}
\item The distinguished candidate $p$;
\item An additional special candidate $a$.
\item A candidate for each clause $c_i$.
\item For each variable $x_i$, a candidate $x_i'$.
\item For each variable $x_i$, a candidate $x_i''$.
\end{itemize}

The auxiliary candidate set $D$ will consist of the following: 

\begin{itemize}
\item The ``literal'' candidates: For each variable $x_i$, candidates $x_i^+$ and $x_i^-$
\end{itemize}

The voter set $V$ will be constructed such that we have the following
relationships between the candidates (both those in $C$ and in $D$):  

\begin{itemize}
\item 
  For every clause $c_i$, candidate $c_i$ beats $p$ by 2 votes.
\item 
  For every variable $x_i$, candidate $x_i'$ beats $p$ by 2 votes.
\item 
  For every variable $x_i$,  $p$ beats candidate $x_i''$ by 2 votes.
\item 
  For every variable $x_i$, candidates $x_i^+$ and $x_i^-$ beat $x_i'$ by
  4 votes. 
\item
  For every variable $x_i$, $x_i''$ beats $x_i^-$, $x_i^-$
  beats $x_i^+$, and $x_i^+$ beats $p$, all by four votes.  
\item 
  For every variable $x_i$, candidates $x_i^+$ beats every clause
  candidate $c_i$ where the corresponding clause is satisfied by $x_i$
  assigned to true, by 4 votes.
\item 
  For every variable $x_i$, candidates $x_i^-$ beats every clause
  candidate $c$ where the corresponding clause is satisfied by $x_i$
  assigned to false, by 4 votes.
\item 
  $p$ beats $a$ by 4 votes.  
\item 
  For every variable $x_i$, $a$ beats candidates $x_i^+$ and $x_i^-$ by
  4 votes.
\end{itemize}

In the limited case, the bound will be equal to $||D||$, the size of
the auxiliary candidate set.


We will show that $p$ can be made a winner through adding candidates
if and only if there is a satisfying assignment for the 3SAT
instance. 

Given a satisfying assignment for the 3SAT instance, build the
corresponding set of added candidates, where $x_i^+$ is included if
$x_i$ is set to true, and $x_i^-$ is included otherwise.  We can show
that $p$ will be a winner with this set of candidates added.  Since
these candidates correspond to a satisfying assignment, and every
literal candidate $x_i^+$ or $x_i^-$ beats by 4 every clause candidate
they satisfy, and since $p$ has a strong path to each literal
candidate (through $a$), there will be  a strong path from $p$ to
each clause candidate $c_i$, stronger than the strength-of-2 edge from
$c_i$ to $p$.  Since we have added exactly one of $x_i^+$, $x_i^-$ for each
$x_i$, we create a strong path from $p$ to the $x_i'$ candidate, and
we prevent there from being a strong path from $x_i''$ to $p$, as such
a path would have to go through both of those candidates.  Thus no
candidate will have a better path to $p$ than $p$ has to that
candidate, so $p$ will be a winner.

We will show that if we map to a positive control instance, we must
have mapped from a satisfiable 3SAT instance.  To make $p$ a winner of
the election, we must ensure $p$'s paths to other candidates are at
least as strong as other candidates' paths to $p$.  $p$ immediately
has strong paths to any of the added literal candidates.  To beat each
of the clause candidates we must add literal candidates that give $p$
a path to each of them.  To beat the $x_i'$ candidates we must add at
least one of $x_i^+$ or $x_i^-$ for each $x_i$.  To beat the $x_i''$
candidates we must not add both of $x_i^+$ and $x_i^-$ for each
$x_i$.  Thus to have a successful control instance, we must set every
variable to one of true or false in such a way that satisfies each
clause.  Thus if we have a positive control instance we must have a
positive 3SAT instance.  
\end{proof}

\subsubsection{Constructive Control by Deleting Candidates Case}

\begin{theorem}
Schulze voting is resistant to constructive control by deleting candidates.
\end{theorem}

\begin{proof}{CC-DC Case}

We prove this case through a reduction from 3SAT.
  Given a 3SAT instance $(U, Cl)$ we will construct a control
  instance $((C,V),p,k)$ as follows. Our deletion limit $k$ will be
  equal to $||Cl||$.  The candidate set will be the following:

\begin{itemize}
\item Our special distinguished candidate $p$.
\item For each clause $c_i \in Cl$, $k+1$ candidates $c_i^1, \dots
  c_i^{k+1}$.
\item A candidate for each literal $x_i^j$ in each of the clauses
  (where $x_i^j$ is the $j$th literal in the $i$th clause).
\item For each pair of literals $x_i^j, x_m^n$ where one is the negation of the
  other, $k+1$ candidates $n_{i,j,m,n}^1, \dots n_{i,j,m,n}^{k+1}$.
\item An additional auxiliary candidate $a$.
\end{itemize}
The $k+1$ candidates for each clause and for each conflicting pair of
literals are treated as copies of each other, and are included to
prevent successful control by simply deleting the tough opponents
rather than solving the more difficult problem corresponding to the 3SAT
instance.  

The voters will be constructed according to the McGarvey method~\cite{mcg:j:election-graph} such
that we have the following relationships between candidates.  
\begin{itemize}
\item For the three literal candidates $x_i^1, x_i^2, x_i^3$ in a clause
  $c_i$, $c_i^j$ beats $x_i^1$ (for all $j$), $x_i^1$ beats $x_i^2$, $x_i^2$ beats $x_i^3$, and
  $x_i^3$ beats $p$, all by two votes.
\item For a pair of literal candidates $x_i^j$ and $x_m^n$ that are
  negations of each other, each beats
  $n_{i,j,m,n}^l$ (for all $l$) by two votes.  
\item Every negation candidate $n_{i,j,m,n}^l$ beats $p$ by two
  votes.  
\item $p$ beats $a$ by two votes.
\item $a$ beats every $x_i^j$ by two votes.  
\end{itemize}


This completes the specification of the reduction, which clearly can
be performed in polynomial time.
The intuition is as follows:  Deleting a literal corresponds to
assigning that literal to be true, and to make $p$ win we must ``assign''
literals that satisfy every clause without ever ``assigning'' a
variable to be both true and false, so a successful deletion
corresponds to a valid satisfying assignment.  

We will now show that if $(U,Cl)$ is a positive 3SAT instance, $p$ can be
made a winner of this election by deleting $k$ candidates.  We will
delete one literal candidate for each clause, selecting a literal that
is satisfied by the satisfying assignment for $Cl$.  This will require
us to delete $||Cl||$ candidates, which is equal to our deletion limit
$k$.  By deleting these candidates, we break all the paths from the
clause candidates to $p$, so now $p$ is tied with each clause
candidate instead of being beaten by them.  Also, since we deleted
literals that were satisfied according to a satisfying assignment, we
must not have deleted any pair of literals that were the negations of
each other, so we will still have a path from $p$ to each negation
candidate.  Thus $p$ at least ties every other candidate in Schulze
score, so they will be a winner.

If we map to a positive control instance, our 3SAT instance must be
positive as well.  First, the most serious rivals to $p$ are the
clause candidates, who each have a path of strength two to $p$ while
$p$ only has a path of strength $0$ to them.  Since there are many
duplicates of each of them, we cannot remove them directly but must
instead remove other vertices along the paths to $p$ to remove the
threat.  The deletion limit allows us to delete one candidate for
every clause.  However, we must choose which ones to delete carefully.
If we delete a literal and a different one that is the negation of
the first, we destroy our only paths to the corresponding negation
candidate.  Thus we must delete one literal for each clause, while
avoiding deleting a variable and its negation.  If there is a
successful way to do this, there will be a satisfying assignment for
the input instance that we can generate using our selection of
deleted/satisfied literals (arbitrarily assigning variables
that were not covered in this selection).
\end{proof}

\subsubsection{Constructive Control by Partition of Candidates Cases}

\begin{theorem}
Schulze voting is resistant to constructive control by partition and
runoff partition of candidates, ties eliminate.
\end{theorem}

\begin{proof}

We will reduce from 3SAT.  Given a 3SAT instance $(Cl, U)$, we
construct a control instance $((C,V), p)$.  The candidate
set $C$ will consist of the following.  

\begin{itemize}
\item The distinguished candidate $p$.
\item For every variable $x_i \in U$, candidates $x_i^+$, $x_i^-$, and
  $x_i'$.  
\item For every clause $c_i \in Cl$, a candidate $c_i$.  
\item  Adversary candidates $c'$ and $u'$.  
\item Helper candidates $a$ and $a'$.  
\end{itemize}

The relationships between candidates will be as follows.  

\begin{itemize}
\item Every clause candidate $c$ beats $p$ by 4 votes.  
\item Every clause candidate $c$ beats $a$ by 4 votes.  
\item $c'$ beats $p$ by 2 votes.
\item $c'$ beats $a$ by 4 votes.
\item $u'$ beats $p$ by 4 votes.  
\item $a$ beats $a'$ by 4 votes.
\item $a'$ beats every $x_i^+$ and $x_i^-$ by 4 votes.  
\item $p$ beats every $x_i^+$ and $x_i^-$ by 2 votes.  
\item $c'$ beats $u'$ by 2 votes.  
\item $c'$ beats every $c_i$ by 2 votes.  
\item $u'$ beats every $x_i'$ by 4 votes.  
\item For every $x_i$, $x_i'$ beats $x_i^-$, $x_i^-$ beats $x_i^+$,
  and $x_i^+$  beats $a$, all by 4 votes.  
\item $a$ beats $u'$ by 2 votes.  
\item For every variable $x_i$, candidates $x_i^+$ beats every clause
  candidate $c_i$ where the corresponding clause is satisfied by $x_i$
  assigned to true, by 4 votes.
\item For every variable $x_i$, candidates $x_i^-$ beats every clause
  candidate $c$ where the corresponding clause is satisfied by $x_i$
  assigned to false, by 4 votes.
\end{itemize}


We will show that $p$ can be made a winner of $(C,V)$ through control
by partition (or runoff partition) of candidates if and only if
$(U,Cl)$ is a positive instance of 3SAT.  

Assume that $(U,Cl)$ is a satisfiable 3SAT instance.  Then consider
the partition of the candidates with the first partition (the one that
undergoes the initial election in the non-runoff case) containing the
$x_i^+$ and $x_i^-$ candidates corresponding to the satisfying
assignment, and also all other candidates except $p$.  We will show
that no candidates will be promoted from this first subelection, and
$p$ will win the final election.  Since we have $x_i^+$ and $x_i^-$
candidates corresponding to a satisfying assignment, there is a strong
path from $a$ to every clause candidate $c_i$.  One of each of $x_i^+$
and $x_i^-$ is not present for every $x_i$, breaking potential paths
from $x_i'$ candidates to $a$.  The overall effect is that $a$ will be
a winner, along with $c'$ and a number of other candidates, and since
we are in the ties-eliminate model, no candidates will be promoted
from this subelection.  Depending on whether we are in the partition
or runoff partition case, $p$ either faces the remaining candidates
(the unchosen variable assignment candidates) in an initial
subelection or in the final election, and they will win.

Assume that we map to an accepting instance of the control problem. We
will show that there must be a satisfying assignment for the 3SAT
instance.  $p$ only beats the variable assignment candidates and $a$,
and only narrowly.  $p$ can't survive an election that contains any of
the other candidates, so we must eliminate them in an initial
subelection.  $c'$ has as path to every candidate of strength 2, and
no other candidate has a path to them of greater than strength 2, so
they will invariably be a winner of any subelection of which they are
a part.  So the only way to eliminate them is to force a tie in an
initial subelection.  $a$ has potential to tie $c'$, but to do so we
have to give them a strong path to every clause candidate, and ensure
there are no paths from any of the $x'$ candidates.  Thus we need to
include variable assignment candidates satisfying each clause, but not
include both the $x^+$ and $x^-$ candidates for any $x$.  So if
control is possible, there must be a valid satisfying assignment for
the 3SAT instance.
\end{proof}

\begin{theorem}
Schulze voting is resistant to constructive control by partition and
runoff partition of candidates, ties promote.
\end{theorem}

\begin{proof}

In the ties-promote case we can use a similar reduction to the
previous case.  In this version, we alter the candidate set and net
advantage scores so that the helper candidate $a$ can be made a unique
winner of their subelection if and only if there is a satisfying
assignment, and $p$ can only win the final election if $a$ is a unique
winner of a subelection.

The candidate set is as follows.  Compared to the previous case, we
have eliminated the adversary candidates $c'$ and $u'$.  They were
important previously to ensure that there was a unique winner other
than $a$ of any subelection not corresponding to a satisfying
assignment, but in this case they are not necessary.

\begin{itemize}
\item The distinguished candidate $p$.
\item For every variable $x_i \in U$, candidates $x_i^+$, $x_i^-$, and
  $x_i'$.  
\item For every clause $c_i \in Cl$, a candidate $c_i$.  
\item Helper candidates $a$ and $a'$.  
\end{itemize}

The new scores will be as follows.

\begin{itemize}
\item $p$ beats $a$ by 2 votes.  
\item Every clause candidate $c$ beats $p$ by 4 votes.  
\item Every clause candidate $c$ beats $a$ by 4 votes.    
\item $a$ beats $a'$ by 6 votes.
\item $a'$ beats every $x_i^+$ and $x_i^-$ by 6 votes.  
\item $a'$ beats every $x_i'$ by 2 votes.
\item $p$ beats every $x_i^+$ and $x_i^-$ by 2 votes.  
\item For every $x_i$, $x_i'$ beats $x_i^-$, $x_i^-$ beats $x_i^+$,
  and $x_i^+$  beats $a$, all by 4 votes.  
\item For every variable $x_i$, candidates $x_i^+$ beats every clause
  candidate $c_i$ where the corresponding clause is satisfied by $x_i$
  assigned to true, by 6 votes.
\item For every variable $x_i$, candidates $x_i^-$ beats every clause
  candidate $c$ where the corresponding clause is satisfied by $x_i$
  assigned to false, by 6 votes.
\end{itemize}


If we map from a positive 3SAT instance we must have a positive
control instance.  The partition can be as follows: In the first
partition include $a$, $a'$, variable assignment candidates
corresponding to a satisfying assignment, and every $c_i$ and $x'$
candidate.  The other partition will contain just $p$ and the
remaining variable assignment candidates.  In the first subelection,
$a$ will be the only winner, as they will have strong paths to all the
clause candidates and no candidate has a strong path to them.  $p$
will win the final election with $a$ and possibly also with some
variable assignment candidates (in the nonrunoff case)

If we map to a positive control instance we must have a satisfiable
3SAT instance.  Since we are working in the TP model, at least one
candidate will survive each subelection.  The only candidates that $p$
beat are $a$ and the variable assignment candidates (by only a narrow
margin), so to win we must use $a$ to eliminate any candidates that
beat $p$ in a subelection.  Thus we have to make sure that $a$ is a
winner of their subelection, and no other candidate that beats $p$ is
a winner.  For this to happen $a$ must have a strong path to every
clause candidate and none of the $x'$ candidates can have strong paths
back to $a$.  We must include a valid satisfying assignment of
variable assignment candidates for this to be the case, so we must
have a satisfiable formula.
\end{proof}

\subsubsection{Destructive Control by Partition of Candidates Cases}

\begin{theorem}
Schulze voting is vulnerable to destructive control by partition of
candidates and destructive control by runoff-partition of candidates
in either the ties-promote or ties-eliminate model.
\end{theorem}

Schulze voting is vulnerable to each of the variants of destructive
control by partition of candidates. We will show this,
and additionally show that these problems are easy for a broad class
of Condorcet voting systems.

In the nonunique-winner model that we are following here, both of the
pairs destructive control by partition of candidates, ties eliminate
and destructive control by runoff partition of candidates, ties
eliminate; and destructive control by partition of candidates, ties
promote and destructive control by runoff partition of candidates,
ties promote are in fact the same problems, that is, the same sets.
This fact was discovered in~\cite{hem-hem-men:t:svd} by noting shared alternative
characterizations of these problems.  The nominal difference is that
in the ``partition of candidates'' case, the candidate set is
partitioned and only one part undergoes a culling subelection while
the other part gets a bye, while in the ``runoff partition of
candidates'' case, both parts of the partition first face a
subelection.  In truth they are much simpler problems, and identical
(within the same tie-handling model).  Namely, the
sets DC-PC-TE and DC-RPC-TE are equivalent to the set of instances
$((C,V), p)$ where there is some set $C' \subseteq C$ with $p \in C'$
such that $p$ is not a unique winner of the election $(C', V)$.  The sets
DC-PC-TP and DC-RPC-TP are equivalent to the set of instances $((C,V),
p)$ where there is some set $C' \subseteq C$ with $p \in C'$ such that
$p$ is not a winner of the election $(C', V)$.  We will build
algorithms for these cases aided by these characterizations.

These algorithms are optimal in a class of voting systems that are Condorcet
voting systems that also possess a weaker version of the
Condorcet criteria---where if there are any candidates that do no
worse than tying in pairwise contests with other candidates (known as weak Condorcet winners), they will
be winners (possibly, but not necessarily unique) of the election.
There can potentially be one or more than one such candidates.
Schulze voting is such a voting system:  If a candidate is a weak
Condorcet winner, doing no worse
than tying against any other candidate, no candidate can have a path to that candidate of strength greater
than 0, and since that candidate at least ties every other candidate,
he or she has a path of strength at least 0 to every other candidate.
Thus he or she is unbeaten in Schulze score, and will be a winner.

Other Condorcet voting systems that possess this property include
Copeland$^{1}$, Minimax, and
Ranked Pairs.

\begin{proof}[DC-PC-TP/DC-RPC-TP Case]
Recall our alternate characterization for these problems, that the set
of positive instances is the same as the set of instances $((C,V), p)$
where there is some set $C' \subseteq C$ with $p \in C'$ such that $p$
is not a winner of the election $(C', V)$.  Thus, finding if we have a
positive instance is very similar to the problem of finding if we have
a positive instance in the destructive control by deleting candidates
problem, except that there is no longer a limit on the number of
candidates we can delete.  Thus we do not have to carefully limit the
number of deleted candidates, but we can freely delete as many
candidates as necessary to put $p$ in a losing scenario.  


Given a control instance $((C,V), p)$, all we must do is see if
there is any candidate $a \in C$, $a \neq p$,  such that $netadv(a,p) > 0$.  If such
a candidate exists, we let our first partition be $\{a,p\}$ and let the
second be $C - \{a,p\}$.  $p$ will lose their initial election to $a$,
and be eliminated from the election. If there is no such candidate,
$p$ is a weak Condorcet winner, and thus they will always be a winner
of the election among any subset of the candidates, as they will be a
weak Condorcet winner (at least, or possible a Condorcet winner) among
any subset of the candidates.  Thus our algorithm will indicate
failure.  
We can perform this check
through a simple examination of the net advantage scores which can be
generated from an election in polynomial time, and so this algorithm
runs in polynomial time.  Thus Schulze voting (and any other system
meeting the aforementioned criteria) is vulnerable and
constructively vulnerable to DC-PC-TP/DC-RPC-TP.
\end{proof}

\begin{proof}[DC-PC-TE/DC-RPC-TE Case]
These cases have the alternate characterization that the set of
positive instances is equivalent to the set of instances $((C,V), p)$
where there is some set $C' \subseteq C$ with $p \in C'$ such that $p$
is not a unique winner of the election $(C', V)$, that is, there are
multiple winners, or no winners, or there is a single winner that
is not $p$.  This case thus differs only slightly from the
DC-PC-TP/DC-RPC-TP case and we can create a very similar simple
algorithm.

Given a control instance $((C,V), p)$, we can find a successful action
if one exists by simply checking if there is any candidate $a \in C$,
$a \neq p$, such that $netadv(a,p) \geq 0$.  If so, we let our first
partition be $\{a,p\}$ and let the second be $C - \{a,p\}$.  This
results in $p$ either not being a winner of the subelection at all or
being a winner along with $a$ if the net advantage score is $0$.
Either way, $p$ will not be promoted to the final election and thus
they will not be a winner of the final election.  If there is no such
candidate, $p$ is a Condorcet winner, and they will be a Condorcet
winner among any subset of the candidates, so this type of control
will never be possible, and our algorithm will indicate that.  As
before, generating the net advantage function is easily possible in
polynomial time, and the simple check will only take polynomial time,
so Schulze voting (and any other system meeting the aforementioned
criteria) is vulnerable and constructively vulnerable to
DC-PC-TE/DC-RPC-TE.
\end{proof}

\subsubsection{Partition of Voters Cases}

\begin{theorem}
Schulze voting is resistant to constructive control by partition 
  of voters, ties eliminate.
\end{theorem}

\begin{proof}
We will reduce from a 3SAT instance $(U, Cl)$ and output a control
instance $((C,V),p)$.  The candidate set $C$ will
be as follows:

\begin{itemize}
\item A distinguished candidate $p$.  
\item A Candidate $c_i$ for every clause $c_i \in Cl$.
\item A candidate $x_i$ for every variable $x_i \in U$.
\item An auxiliary candidate $a$.
\end{itemize}

The voter set $V$ will be as follows:  


\begin{itemize}
\item For each variable $x_i$, where $D$ is the set of clauses
  satisfied by $x_i$ assigned to true, a voter ranking $p$ over $c_i$
  for $c_i \in D$ and $c_j$ over $p$ for $c_j \notin D$, ranking $x_i$
  over $p$, but ranking $p$ over $x_j$ for $x_j \in U - \{x_i\}$, and
  ranking $p$ over $a$.

\item For each variable $x_i$, where $D$ is the set of clauses
  satisfied by $x_i$ assigned to false, a voter ranking $p$ over $c_i$
  for $c_i \in D$ and $c_j$ over $p$ for $c_j \notin D$, ranking $x_i$
  over $p$, but ranking $p$ over $x_j$ for $x_j \in U - \{x_i\}$, and
  ranking $p$ over $a$.

\item $||U|| - 3$ voters preferring $p$ over every $c_i$, but preferring
  every $x_i$ and $a$ over $p$.  

\item $1$ voter preferring $p$ over every $c_i$ and over $a$, but
  preferring every $x_i$ over $p$.

\item $2$ voters preferring $p$ over every $c_i$ and $x_i$ but with $a$
  preferred over $p$.  
\end{itemize}


If we are mapping from a positive 3SAT instance, it will be possible
to make $p$ win the final election through this type of control.  The
partition will be as follows: The first partition will contain all of
the third, fourth, and fifth groups of voters, and will contain voters
from the first two groups corresponding to a satisfying assignment.
The result of adding the voters from the third, fourth, and fifth
groups will be to give $p$ $||U||$ votes over each $c_i$, each $x_i$
$||U||-4$ votes over $p$, and $a$ $||U||-2$ votes over $p$.  By adding
the variable assignment voters corresponding to satisfying assignment,
we decrease every $p$-$c_i$ edge by no more than $||U||-2$, increase
every $p$-$x_i$ edge by $||U||-2$, and increase the $p$-$a$ edge by
$||U||$.  Thus $p$ will beat every candidate and win their
subelection.  In the other subelection, with the remaining variable
assignment candidates, $p$ will tie with some of the clause candidates
(unless the inverse assignment is also satisfiable) and so no
candidates will be promoted.  $p$ will thus be alone in the final
election and win.

If we map to a positive control instance, the 3SAT instance will be
satisfiable.  No voter prefers $p$ outright, so we must balance their
votes against the other voters to have a chance.  The voters in the
first two groups prefer $p$ over $a$ and they mostly prefer $p$ over
the variable candidates, but they mostly prefer the clause candidates
over $p$.  The voters in the third, fourth, and fifth groups can boost
$p$ over the clause candidates, but they give $a$ and the variable
candidates an advantage over $p$.  To make $p$ the only winner of a
subelection, we have to carefully select the voters from the first two
groups to keep $p$ ahead of both the variable and clause candidates.
To do so we must ensure that at least one voter prefers $p$ over each
$c_i$, and only one voter prefers each $x_i$ over $p$.  These voters
thus correspond to a satisfying assignment, so the 3SAT instance must
be satisfiable.  
\end{proof}

\begin{theorem}
Schulze voting is resistant to constructive control by partition 
  of voters, ties promote.
\end{theorem}

\begin{proof}
This case can be shown through a reduction from 3SAT very similar to
the previous case.  We will map from a 3SAT instance $(U, Cl)$ to a
control instance $((C,V),p)$.  We will have the following candidate
set $C$:

\begin{itemize}
\item A distinguished candidate $p$.  
\item A Candidate $c_i$ for every clause $c_i \in Cl$.
\item A candidate $x_i$ for every variable $x \in U$.
\item An auxiliary candidate $a$.
\end{itemize}

The voter set $V$ will be the following:

\begin{itemize}
\item For each variable $x_i$, where $D$ is the set of clauses
  satisfied by $x_i$ assigned to true, a voter ranking $p$ over $c_i$
  for $c_i \in D$ and $c_j$ over $p$ for $c_j \notin D$, ranking $x_i$
  over $p$, but ranking $p$ over $x_j$ for $x_j \in U - \{x_i\}$, and
  ranking $p$ over $a$.

\item For each variable $x_i$, where $D$ is the set of clauses
  satisfied by $x_i$ assigned to false, a voter ranking $p$ over $c_i$
  for $c_i \in D$ and $c_j$ over $p$ for $c_j \notin D$, ranking $x_i$
  over $p$, but ranking $p$ over $x_j$ for $x_j \in U - \{x_i\}$, and
  ranking $p$ over $a$.

\item $||U|| - 1$ voters preferring $p$ over every $c_i$, but preferring
  every $x_i$ and $a$ over $p$.  

\item $1$ voter preferring $p$ over every $c_i$ but with $a$ and
  every $x_i$ preferred over $p$.

\item $||U|| + 2$ voters preferring $a$ over the variable assignment
  candidates over the clause candidates over $p$.   
\end{itemize}

We have just changed the later groups of voters and left the variable
assignment voters the same.  Now, with just the third and fourth
groups of voters present $a$ has $||U||$ votes over every $c_i$ and
$||U||-2$ votes over every $x_i$, and $a$ has $||U||-2$ votes over $p$.
This proof can proceed similarly to the previous case, but noting that
now we are in the TP model, multiple candidates can be promoted to the
final election, so we must ensure that no candidates that can beat $p$
with the entire voter set will make it.  Thus to make $p$ a winner we
must cause them to at least tie the $x_i$ and $a$ candidates, and cause them
to beat the $c_i$ candidates in the first round.  For this to
happen we must include in a partition voters corresponding to a
satisfying assignment, as well as the third and fourth groups of
voters.  In the other partition we must prevent any of the clause
candidates from winning, and this will happen with the remaining
variable assignment voters and the fifth group of voters: $a$ will be
the only winner.  Thus $p$ will end up being a winner of the final
election.  

If we have a positive control instance, $p$ must be a winner of at
least one initial subelection, and so this partition must have
balanced the points $p$ gets over the clause candidates from the third
and fourth groups with the points $p$ gets over the variable
candidates and $a$ from the first two groups, and the voters from the
first two groups must have been carefully chosen so at least one of
them has $p$ over each clause candidate, and no more than one has each
variable candidate over $p$.  Thus there must be a satisfying
assignment for the 3SAT instance.
\end{proof}

\begin{theorem}
Schulze voting is resistant to destructive control by partition 
  of voters, ties eliminate.
\end{theorem}

\begin{proof}

We will reduce from a 3SAT instance $(U, Cl)$ and construct a control
instance $((C,V), p)$.  The candidate set $C$ will be as follows:

\begin{itemize}
\item The distinguished candidate $p$.
\item The adversary candidates $a$ and $b$. 
\item Candidates $c_i$ and $c_i'$ for each clause $c_i \in Cl$.
\item Candidates $x_i$ and $x_i'$ for each variable $x_i \in U$. 
\end{itemize}

The voter set $V$ will be as follows: 

\begin{itemize}
\item For each variable $x_i$, where $D$ is the set of  clauses
  satisfied by $x_i$ assigned to true, a voter ranking $c_i'$ over $c_i$ for
  $c_i \in D$ and $c_j$ over $c_j'$ for $c_j \in Cl - D$, ranking $x_i'$ over
  $x_i$, but ranking $x_j$ over $x_j'$ for $x_j \in U - \{x_i\}$, ranking $p$
  over $a$ and $b$, the clause candidates over $a$ and the variable
  candidates over $b$.  

\item For each variable $x_i$, where $D$ is the set of  clauses
  satisfied by $x_i$ assigned to false, a voter ranking $c_i'$ over $c_i$ for
  $c_i \in D$ and $c_j$ over $c_j'$ for $c_j \in Cl - D$, ranking $x_i'$ over
  $x_i$, but ranking $x_j$ over $x_j'$ for $x_j \in U - \{x\}$, ranking $p$
  over $a$ and $b$, the clause candidates over $a$ and the variable
  candidates over $b$.

\item $2||U||-2$ voters ranking $a > p > b$, with the other candidates
  evenly distributed.
\item $2||U||-2$ voters ranking $b > p > a$,  with the other candidates
  evenly distributed.
\end{itemize}

The  unspecified parts of the votes should  be set to make the 
groups of candidates as even as possible, except for the $c_i'$
candidates, among whom there should be strong paths from each
candidate to each candidate.

If the mapped-from 3SAT instance is satisfiable, we will be able to
defeat $p$ through partition of voters.  The partition we use will be
the following.  The first part will consist of voters from the first
two groups corresponding to a satisfying assignment, and the voters
from the third group.  The second part will consist of the remaining
voters from the first two groups, and the voters from the fourth
group.  $a$ will defeat $p$ in the first partition.  By selecting the
voters to correspond to a satisfying assignment we weaken all of the
paths from $p$ to $a$ going through the clause candidates to be
$||U||-2$ in strength.  The third group of voters provide support for
$a$ over $p$, bringing the $a \rightarrow p$ edge to be $||U||-2$ in
strength.  Thus $a$ ties $p$, and no candidate will be promoted from
this subelection.  In the other subelection, $b$ will tie $p$ for
similar reasons, as they will have an edge of $||U||-2$ to $p$ while
that will also be the strength of the best path from $p$ to $b$ going
through any of the variable candidates.  Thus no candidate makes it to
the final election and control is successful.

If the mapped-to control instance is positive, we must have a positive
3SAT instance.  $p$ will win the final election against any other
candidate if they make it there, so to defeat them we must make them
lose both initial subelections.  No voters rank $p$ below both $a$ and
$b$, the strongest other candidates, so we must defeat $p$ with $a$ in
one subelection and defeat $p$ with $b$ in the other.  To make $b$
beat $p$, we have to include the fourth group of voters, the only
voters that rank $b$ over $p$.  $b$ would defeat $p$ if we include
just these voters, but then $p$ would win the other subelection with
the support of all the variable assignment voters.  Thus we have to
include half the variable assignment voters to give $a$ a chance
as well, and we must pick voters to weaken the paths from $p$ to $b$.
This will require us to pick only one voter corresponding to each
variable.  The other subelection will proceed similarly, but we will
have to pick at least one voter satisfying each clause to weaken the
$p$-$a$ paths.  Thus there must be a valid satisfying assignment for the
3SAT instance if control is possible.
\end{proof}

\begin{theorem}
Schulze voting is resistant to destructive control by partition 
  of voters, ties promote.
\end{theorem}

\begin{proof}
This case can proceed very similarly to the previous case in the TE
model, except that we just have to slightly change some of the numbers
so that the adversary candidates each will fully defeat rather than
tie with $p$ in order to eliminate $p$ from both initial subelections.
The new voters will be as follows, and everything else can proceed in
the same way.

\begin{itemize}
\item For each variable $x_i$, where $D$ is the set of  clauses
  satisfied by $x_i$ assigned to true, a voter ranking $c_i'$ over $c_i$ for
  $c_i \in D$ and $c_j$ over $c_j'$ for $c_j \in Cl - D$, ranking $x_i'$ over
  $x_i$, but ranking $x_j$ over $x_j'$ for $x_j \in U - \{x_i\}$, ranking $p$
  over $a$ and $b$, the clause candidates over $a$ and the variable
  candidates over $b$.  

\item For each variable $x_i$, where $D$ is the set of  clauses
  satisfied by $x_i$ assigned to false, a voter ranking $c_i'$ over $c_i$ for
  $c_i \in D$ and $c_j$ over $c_j'$ for $c_j \in Cl - D$, ranking $x_i'$ over
  $x_i$, but ranking $x_j$ over $x_j'$ for $x_j \in U - \{x_i\}$, ranking $p$
  over $a$ and $b$, the clause candidates over $a$ and the variable
  candidates over $b$.  

\item $2||U||$ voters ranking $a > p > b$, with the other candidates
  evenly distributed.
\item $2||U||$ voters ranking $b > p > a$,  with the other candidates
  evenly distributed.
\end{itemize}

\end{proof}

\subsubsection{Destructive Adding/Deleting Candidates}

For the final cases of control, we do not resolve their complexity, but
we do reduce the problems to their core:  They are no harder than a
simpler problem called path-preserving vertex cut.  

\begin{proposition}
Each of destructive control by adding candidates, destructive control
by unlimited adding of candidates, and destructive control by deleting
candidates in Schulze voting polynomial-time Turing reduces to \emph{path-preserving
  vertex cut}.
\end{proposition}
\begin{proof}

We define path-preserving vertex cut as follows.

\begin{description}
\item[Given] A directed graph $G = (V,E)$, distinct vertices $s,t \in
  V$, and a deleting limit $k \in \mathbb{N}$.  
\item[Question]
  Is there a set of vertices $V'$, $||V'|| \leq k$, such that the
  induced graph on $G$ with $V'$ removed contains a path from $t$ to
  $s$ but not any path from $s$ to $t$?
\end{description}

We note that the standard vertex cut problem is well-known to be solvable
in polynomial time, but to the best of our knowledge the complexity of
this variant is unknown (other than that it is clearly in NP).  Given
this result, Schulze voting will be vulnerable to each of these
control cases if this problem is found to be in P\@.

We will describe how these control cases can reduce to this problem,
first handling the DC-DC case.  Our input will be a DC-DC instance
$((C,V), p, k)$.  Since this destructive case of control, we only have
to make the distinguished candidate $p$ not a winner rather than
making any particular candidate positively a winner.  To do this we
must alter the election so that some other candidate has a better path to
$p$ than $p$ has to that candidate.  Since this is a case of control
by deleting candidates, we can only eliminate paths in the election
graph, not create them.  So we have to try to succeed by breaking the
paths from $p$ to some other candidate.  

We can first handle the cases where $p$ is a Condorcet winner (control
will be impossible) or where $p$ is already beaten by some candidate
(control is trivial).  Otherwise, we will loop through the candidates
and see if we can make any of them beat $p$.  

For each candidate $a$, we will first find the subgraph containing all
maximum-strength paths to $p$ (using a modified Floyd-Warshall
algorithm), and also the subgraph containing all paths from $p$ to $a$
at least as strong as the strongest $a$-$p$ path.  We will then take
the union of these graphs, disregard the edge weights, and apply a
subroutine for path-preserving vertex cut.  This will tell us if there
is any way to cut all of the strong paths from $p$ to $a$ while
preserving one of the strongest ones from $a$ to $p$.  If we ever get
a positive result from this call, we can indicate success.  Now, it
may be that we can succeed by preserving a path other than the
strongest path
from $a$ to $p$.  So if we failed, we repeat the process except
considering paths from $a$ to $p$ and from $p$ to $a$ that are a
little less strong, going down as far as the strength of the $p$-$a$
edge (which is not cuttable).  This loop is polynomially bounded in
the number of edges, as there can only be as many path strengths as
there are edges.  We indicate failure if we never achieve success with
any vertex or path strength.

The other cases can be handled similarly.  The difference is that they
are adding-candidates problems instead of deleting-candidates
problems, but we can reduce them to modified deleting-candidates
problems by including all of the addable vertices and then solving
them as deleting-candidates problems where we can only delete the
vertices in the added sets.  Additionally, in the non-unlimited case,
we would be restricted as to how many of the vertices are added, so we
would have to only look at paths from $a$ to $p$ that use a limited
number of the added vertices.
\end{proof}


\bibliography{grycurtis}
\bibliographystyle{alpha}

\end{document}